\documentclass[letterpaper, 10 pt, conference]{ieeeconf}  %

\IEEEoverridecommandlockouts                              %
\usepackage{mystyle}

\newcommand{\setA}{\mathbf{A}}
\newcommand{\tf}{t_{\rm f}}
\newcommand{\setLambda}{\mathbf{\Lambda}}
\newcommand{\setSigmaD}{\mathbf{\Sigma}_{\mathcal{D}}}
\newcommand{\setSigmapk}{\mathbf{\Sigma}_{\rm pk}}
\newcommand{\setSigma}{\mathbf{\Sigma}}
\DeclareMathOperator{\affin}{aff}
\DeclareMathOperator{\image}{Im}
\DeclareMathOperator{\kernel}{Ker}

\title{\LARGE \bf
  Informativity of Data-Knowledge Pairs for Lyapunov Equations
}

\author{Ikumi Banno%
\thanks{This work was supported by
  the JST Moonshot R\&D Grant Number JPMJMS2021 and
  JSPS KAKENHI Grant Number JP25K17627}%
\thanks{I. Banno is with the Graduate School of Informatics, 
  Kyoto University, Kyoto 606-8501, Japan
        {\tt\small banno@i.kyoto-u.ac.jp}}%
}

\begin{document}
\maketitle
\thispagestyle{empty}
\pagestyle{empty}
\begin{abstract}
  In the past few years, 
  data informativity with prior knowledge has attracted increasing attention. 
  This line of research aims to characterize
  whether data and prior knowledge suffice for system analysis or design.
  In this paper, we investigate such a characterization
  for the data-driven problem of determining a unique solution to Lyapunov equations.
  First, we introduce a notion of \textit{joint informativity}
  for data-knowledge pairs as an extension of the standard informativity concept.
  Second, we derive
  an algebraic necessary and sufficient condition for the joint informativity.
  Finally, we provide further insights into the joint informativity 
  by considering a special case of prior knowledge.
  The characterization presented in this paper
  is developed for a wide class of prior knowledge, 
  enabling the incorporation of various forms of system information.
\end{abstract}

\section{INTRODUCTION}
\label{sec:introduction}

Data-driven approaches
\cite{bib:hou2013,bib:ljung1987,bib:vanwaarde2025}
have emerged as a promising paradigm
for the analysis and control of real-world dynamical systems.
In many practical situations, the system dynamics
involve unknown parameters or may be entirely unknown, 
which makes it difficult to directly apply standard model-based techniques.
To overcome this limitation, 
data-driven approaches utilize measurement data
for system analysis and control
without requiring a known model a priori.
These frameworks can be broadly classified into two categories: 
\textit{indirect methods} and \textit{direct methods}
\cite{bib:vanwaarde2025}. 
The former are based on system identification followed by
standard model-based methods,
whereas the latter analyze or design dynamical systems
directly from data, bypassing explicit system identification.

In the analysis and design of physical systems, 
incorporating prior knowledge on the system dynamics
into a data-driven framework plays an important role
\cite{bib:ljung1987}. 
One advantage is that 
it can help reduce the amount of data required for data-driven tasks.
Another advantage is that
the incorporation of prior knowledge can improve 
the reliability and explainability of the obtained results. 
For these reasons, data-driven methods that incorporate 
prior knowledge have a long history of study.
In particular, such approaches have been well developed
in the context of indirect methods, 
where \textit{gray-box identification} 
\cite{bib:ljung1987} has been established. 
In contrast, direct data-driven methods have mainly
been developed for \textit{black-box} settings
\cite{bib:hjalmarsson1998IFT,bib:persis2020}, 
where prior knowledge is not incorporated. 
However, some works have considered 
direct methods with prior knowledge, 
e.g., 
network reconstruction \cite{bib:vanwaarde2019},
robust control \cite{bib:berberich2023},
and model predictive control \cite{bib:watson2025hybrid}.

Recently, \textit{data informativity} has been proposed
as a novel framework for direct data-driven methods
\cite{bib:vanwaarde2020,bib:vanwaarde2023,bib:vanwaarde2025}. 
This concept characterizes
whether a given dataset contains sufficient information
to solve a data-driven problem of interest. 
Within this paradigm, conditions under which 
data are informative for various analysis and control tasks
have been investigated
\cite{bib:vanwaarde2020,bib:vanwaarde2020suboptimal,bib:vanwaarde2023,
  bib:trentelman2022,
  bib:bannoLCSS2023,bib:bannoCDC2025}.
In particular, the author has previously studied Lyapunov equations,
a fundamental framework for system analysis and control, in this context
\cite{bib:bannoLCSS2023,bib:bannoCDC2025}.
Solving these equations is fundamental to stability analysis
and the computation of controllability and observability Gramians\cite{bib:bernstein2018}.
The results on data informativity discussed above
characterize when measurement data alone are sufficient
to solve a given system-theoretic problem.
However, in many situations, additional information on the underlying system
is available beyond measurement data.
In such cases, lack of informativity of a dataset does not necessarily imply
that the problem is unsolvable.
Indeed, prior knowledge may complement the available data
and make the problem solvable even when the dataset itself is not informative.
This motivates the need to characterize data-knowledge pairs, 
rather than datasets alone, that enable system analysis and design.

From this viewpoint,
recent studies have considered
data informativity with prior knowledge
\cite{bib:iwata2025,bib:shakouri2025,bib:huang2025,bib:shakouri2026}.
In particular, \cite{bib:iwata2025,bib:shakouri2025}
consider prior knowledge on system-theoretic properties:
\cite{bib:iwata2025} characterizes datasets enabling stability analysis
under the assumption of system positivity,
while \cite{bib:shakouri2025} studies data-driven stabilization
under prior knowledge such as stabilizability or controllability.
In contrast, \cite{bib:huang2025} considers a different type
of prior knowledge, namely that some entries of the system
matrices are exactly known, and applies it to
data-driven stabilization of polynomial systems.
More recently, \cite{bib:shakouri2026} investigates
experiment design for set-membership identification
under general sets of prior knowledge.
However, to the best of our knowledge,
a characterization of data-knowledge pairs
for determining the unique solution to Lyapunov equations
has not yet been established.

Therefore, this paper addresses this problem
for general classes of prior knowledge.
First, we introduce an informativity notion for data-knowledge pairs for this problem.
This notion, referred to as \textit{joint informativity} in this paper,
extends the standard data informativity concept in \cite{bib:vanwaarde2020}. 
Next, we derive an algebraic necessary and sufficient condition
for joint informativity.
The key observation underlying this derivation is that
the existence of a common Lyapunov solution is preserved
when the set of state matrices consistent with the data
and prior knowledge is replaced by its affine hull.
Using a finite affine representation of this hull,
we characterize the joint informativity in terms of the existence
of a common solution to finitely many Lyapunov equations.
Based on this characterization,
we further develop a data-driven method for solving Lyapunov equations. 
Finally, we provide additional insights into joint informativity
by specializing the analysis to a structured class of prior knowledge. 
The proposed framework is illustrated through numerical examples.

This paper is organized as follows.
\secref{sec:prob} introduces the notion of joint informativity and then
formulates data-driven problems for Lyapunov equations.
Solutions to the problems are presented in \secref{sec:solution}.
\secref{sec:conclusion} concludes the paper.

\textbf{Notation}:
\textit{(i) Standard sets}:
$\setR$, $\setRp$, $\setZp$, and $\setZzp$ denote
the sets of real numbers, positive real numbers, 
positive integers, and nonnegative integers, respectively.

\noindent
\textit{(ii) Matrices and vectors}:
For a matrix $M \in \setR^{n \times m}$,
$M^+ \in \setR^{m \times n}$, $\image(M) \subseteq \setR^n$, and $\kernel(M) \subseteq \setR^m$
are the Moore-Penrose generalized inverse, the column space, and the null space of $M$, respectively.
For a square matrix $A \in \setR^{n \times n}$,
we denote the set of all eigenvalues of $A$ by $\setLambda(A)$.
The $n \times n$ identity matrix and
the $n \times m$ zero matrix are denoted by $I_n$ and $0_{n \times m}$, respectively.
When the dimensions of a zero matrix are clear
from the context, we simply denote it by $0$.

\noindent
\textit{(iii) Vector and affine subspaces}:
Let $\mathbf{V}$ be a real vector space. %
For an element $a \in \mathbf{V}$ and
a vector subspace $\mathbf{W} \subseteq \mathbf{V}$,
we write
$a + \mathbf{W} \coloneqq \{a + w \mid w \in \mathbf{W}\}$
for the corresponding affine subspace.
For a set $\mathbf{S} \subseteq \mathbf{V}$,
$\vspan(\mathbf{S})$ and $\affin(\mathbf{S})$
denote the linear span and the affine hull of
$\mathbf{S}$, respectively.

\noindent
\textit{(iv) Empty matrices}:
Consider matrices
$X \in \setR^{n \times l}$,
$Y \in \setR^{l \times l}$, and
$Z \in \setR^{l \times m}$.
When $l=0$, we adopt the conventions
$XZ=XYZ=0_{n \times m}$ and $\image(X)=\{0_{n \times 1}\}$.
The generalized inverse of an empty matrix
$M \in \setR^{p \times q}$ is the unique empty matrix
in $\setR^{q \times p}$.
Any matrix with zero columns is
regarded as having full column rank.
For example, when $l=0$,
we have $X^+ \in \setR^{0 \times n}$ and 
$XX^+ = 0_{n \times n}$,
and $X$ has full column rank.

\section{PROBLEM FORMULATION}
\label{sec:prob}

\subsection{General Concept of Joint Informativity}
  This subsection introduces a general framework for \textit{joint informativity}
  for pairs of data and prior knowledge of dynamical systems.
  
  Let $\mathcal{M}$ be a model class of dynamical systems
  and $\mathcal{S} \in \mathcal{M}$ the system of interest for data-driven analysis.
  Here, we assume that $\mathcal{S}$ is unknown. 
  Instead, we are given two sources of information: 
  (i) a dataset $\mathcal{D}$ obtained from measurements of the system behavior,  and
  (ii) prior knowledge represented by a set $\setSigmapk \subseteq \mathcal{M}$,
  which contains the true system $\mathcal{S}$ and
  captures a priori information about $\mathcal{S}$ derived from, 
  e.g., physical laws, structural properties, or parameter bounds.

  We are interested in determining 
  whether the available data and prior knowledge  guarantee that the true system $\mathcal{S}$ satisfies a given property $\mathcal{P}$.
  To this end, let $\setSigma_{\mathcal{P}} \subseteq \mathcal{M}$ be the set of systems satisfying the property $\mathcal{P}$, 
  and let $\setSigmaD \subseteq \mathcal{M}$ be the set of systems consistent with the dataset $\mathcal{D}$.
  Then, our situation is illustrated by \figref{fig:JI}.  
  Throughout the paper,
  we do not require the available information to uniquely identify the system; 
  that is, we do not assume that $\setSigmaD \cap \setSigmapk = \{\mathcal S\}$.

  \begin{figure}[t]
    \centering
    \vspace{2.5mm}
    \definecolor{SageGreen}{HTML}{B4C9B8}
    \definecolor{LightGray}{HTML}{EBEBEB}
    \begin{tikzpicture}[x=1pt,y=1pt]
      
      \filldraw[fill=white,thick] (0,0) rectangle (230,120);
      \node[below left=3pt] at (230, 120) {$\mathcal{M}$};
      
      \filldraw[
        fill=LightGray,
        rounded corners=4pt,
        thick
        ] (85,10) rectangle (145,100);
      \node[below left=2pt and 1pt] at (145,100) {$\mathbf{\Sigma}_{\mathcal{P}}$};
      
      \begin{scope}
        \clip (80,50) ellipse [x radius=55, y radius=30];
        \fill[SageGreen]
        (150,50) ellipse [x radius=55, y radius=30];
      \end{scope}
      
      \draw[thick]
      (80,50) ellipse [x radius=55, y radius=30];
      \node at (65,60) {$\setSigmaD$};
      
      \draw[thick]
      (150,50) ellipse [x radius=55, y radius=30];
      \node at (165,60) {$\setSigmapk$};
      
      \fill (110,53) circle (2pt);
      \node[right=2pt] at (110,53) {$\mathcal S$};   
   
    \end{tikzpicture}
    \caption{Illustration of joint informativity. 
      In this case, the data-knowledge pair
      $(\mathcal{D}, \setSigmapk)$ is jointly informative for the property $\mathcal{P}$.}
    \label{fig:JI}
  \end{figure}
  
  We can deduce that the true system $\mathcal{S}$ has the property $\mathcal{P}$
  from the pair $(\mathcal{D}, \setSigmapk)$ if and only if
  all the systems in $\setSigmaD \cap \setSigmapk$ have the property $\mathcal{P}$.
  This naturally motivates us to introduce the following concept.

  \begin{definition}
    \label{def:JI}
    We say that
    \textit{%
      $(\mathcal{D}, \setSigmapk)$ is jointly informative
      for the system property $\mathcal{P}$%
      }
      if $\setSigmaD \cap \setSigmapk \subseteq \setSigma_{\mathcal{P}}$ holds.    
    \thmend
  \end{definition}
  
  Several remarks on \defref{def:JI} are given, 
  with reference to the notion of data informativity proposed in \cite{bib:vanwaarde2020}. 
  
  First, \defref{def:JI} generalizes the notion of data informativity. 
  Indeed, when $\setSigmapk = \mathcal{M}$, 
  i.e., when no prior knowledge is available, 
  joint informativity coincides with data informativity. 
  
  Second, joint informativity clarifies the benefit of
  incorporating prior knowledge into data-driven analysis. 
  Consider the situation illustrated in \figref{fig:JI}. 
  In this case, the dataset $\mathcal{D}$ is not informative for the property $\mathcal{P}$, 
  i.e., $\setSigmaD \subseteq \setSigma_{\mathcal{P}}$ does not hold. 
  Nevertheless, the pair $(\mathcal{D},\setSigmapk)$ is jointly informative for $\mathcal P$. 
  This example demonstrates that 
  system analysis may still be possible even when existing frameworks based solely on data are not applicable. 
  This advantage becomes particularly evident in the extreme case $\setSigmaD = \mathcal{M}$.

  Third, joint informativity can also be employed for controller design
  to guarantee that the closed-loop system satisfies a given property $\mathcal{P}$. 
  Such an extension can be developed along the same lines as in \cite{bib:vanwaarde2020}. 
  We do not pursue this direction further in the present paper.

  Finally, the term joint informativity emphasizes
  that informativity depends on both the data
  and the prior knowledge.
  Related formulations have been studied for stabilization
  \cite{bib:iwata2025,bib:shakouri2025,bib:huang2025}. 
  In this paper, we derive an algebraic necessary
  and sufficient condition for joint informativity
  for Lyapunov equations under general sets of prior knowledge.
  The definition admits a natural interpretation 
  even when $\mathcal{D}$ itself contains no information (i.e., $\setSigmaD = \mathcal{M}$), 
  since informativity is then determined entirely by the prior knowledge $\setSigmapk$.

\subsection{Joint Informativity for Lyapunov Equations}
  This paper focuses on the joint informativity for solving Lyapunov equations.
  The problems addressed in the paper are formulated as follows.
  
  Consider the linear system
  \begin{equation}
    \label{eqn:sys}
    \dot{x}(t) = Ax(t),
  \end{equation}
  where $x(t) \in \setR^{n}$ is the state of the system
  and $A \in \setR^{n \times n}$ is a constant matrix.  
  For this system,
  a Lyapunov equation is given as follows:
  \begin{equation}
    \label{eqn:Lyap}
    A P + P A^\top = -Q,
  \end{equation}
  where 
  $Q \in \setR^{n \times n}$ is a constant matrix and
  $P \in \setR^{n \times n}$ is the unknown of the equation.
  The Lyapunov equation \eqref{eqn:Lyap} has a unique solution $P$
  if and only if
  \begin{equation}
    \label{eqn:Lyap_eigenvalue_cond}
    \forall \lambda, \lambda^\prime \in \setLambda(A) \ \ 
    \lambda + \lambda^\prime \neq 0
  \end{equation}
  holds \cite{bib:bernstein2018}.  
  
  Throughout the paper, 
  we assume $A \in \setA_n$, where
  \begin{equation}
    \label{eqn:setAn}
    \setA_n \coloneqq \{ A \in \setR^{n \times n} \mid \text{\eqref{eqn:Lyap_eigenvalue_cond} holds}\},
  \end{equation}
  to guarantee the existence of
  the unique solution to the Lyapunov equation \eqref{eqn:Lyap}.
  In the following,
  $\Phi(A, Q) \in \setR^{n \times n}$
  denotes the unique solution $P$ to \eqref{eqn:Lyap}
  for the pair $(A, Q) \in \setA_n \times \setR^{n \times n}$.
  
  Based on the above preparations,
  let us introduce our informativity concept for
  solving Lyapunov equations in \eqref{eqn:Lyap}.
  We consider the case $\mathcal{M} = \setA_n$.
  For a given $x_0 \in \setR^n$,
  let $x(t, x_0) \in \setR^n$ be the solution to \eqref{eqn:sys}
  under the initial condition $x(0) = x_0$.
  Then, the dataset $\mathcal{D}$ is defined as
  a fragment of the state trajectory $x(t, x_0)$, i.e.,
  \begin{equation}
    \label{eqn:def_dataset}
    \mathcal{D} \coloneqq \bigcup_{t \in [0, \tf)} \{(t, x(t, x_0))\},
  \end{equation}
  where $\tf \in \setRp \cup \{\infty \}$ is the duration of the observation interval.
  For the dataset $\mathcal{D}$,
  the set of data-consistent systems is given by
  \begin{equation}
    \label{eqn:def_setSigmaD}
    \setSigmaD \coloneqq \{ \tilde{A} \in \setA_n \mid
    \forall (\tau, \xi) \in \mathcal{D} \ \ \xi = e^{\tilde{A} \tau}x_0 \}.
  \end{equation}
  Meanwhile, 
  $\setSigmapk \subseteq \setA_n$
  is given as a candidate set of state matrices of the system \eqref{eqn:sys},
  satisfying
  \begin{equation}
    \label{eqn:A_in_setSigmapk}
    A \in \setSigmapk.
  \end{equation}
  
  Then, the joint informativity addressed in this paper is formalized as follows.
  
  \begin{definition}
    \label{def:JI_Lyap}
    Consider the system \eqref{eqn:sys}, 
    where we assume $A \in \setA_n$.
    For this system, suppose that
    a dataset $\mathcal{D}$ as in \eqref{eqn:def_dataset}, 
    prior knowledge $\setSigmapk \subseteq \setA_n$ satisfying \eqref{eqn:A_in_setSigmapk},
    and a matrix $Q \in \setR^{n \times n}$ are given.
    Then, we say that
    $(\mathcal{D}, \setSigmapk)$ is 
    \textit{%
      jointly informative
      for the Lyapunov equation (\ref{eqn:Lyap})%
      }
    if
    \begin{equation}
      \label{eqn:Phi=Phi}
      \Phi(\tilde{A}_1, Q) = \Phi(\tilde{A}_2, Q)
    \end{equation}
    holds for all 
    $\tilde{A}_1, \tilde{A}_2 \in \setSigmaD \cap \setSigmapk$.
     \thmend
  \end{definition}
  
  In \defref{def:JI_Lyap},
  $A \in \setA_n$ is assumed 
  to guarantee that the Lyapunov equation \eqref{eqn:Lyap} has a unique solution.
  This restriction can be regarded as additional prior knowledge 
  for the system \eqref{eqn:sys}.
  Note that \defref{def:JI_Lyap} corresponds to \defref{def:JI} for the case
  $\setSigma_{\mathcal{P}} = \{ \tilde{A} \in \setA_n \mid \Phi(\tilde{A}, Q) = \Phi(A, Q) \}$.

  \begin{example}
    \label{ex:motivating_example}
    Consider the system \eqref{eqn:sys}
    and the Lyapunov equation \eqref{eqn:Lyap}
    with
    \begin{equation}
      \label{eqn:ex/A_and_Q}
      A = \begin{bmatrix}
        -1 & 1 \\
        0 & -2
        \end{bmatrix}, \ \ 
      Q = \begin{bmatrix}
        2 & 3 \\
        -3 & 0
      \end{bmatrix},
    \end{equation}
    where $A \in \setA_2$.
    In this case, the unique solution to the Lyapunov equation is given by
    \begin{equation}
      \label{eqn:ex/PhiAQ}
      \Phi(A,Q) = 
      \begin{bmatrix}
        1 & 1 \\
        -1 & 0
      \end{bmatrix}.
    \end{equation}
    Suppose that the dataset $\mathcal{D}$ is given by
    \eqref{eqn:def_dataset} for $\tf = 2$ and $x_0 = [1 \ \ 0]^\top$
    (i.e., $x(t) = [e^{-t} \ \ 0]^\top$)
    and the prior knowledge $\setSigmapk$ is given by
    \begin{equation}
      \label{eqn:ex/setSigmapk}
      \setSigmapk = \{\tilde{A} \in \setA_2 \mid
        0 < \tilde{a}_{12} < 2, \ \ 
        \tilde{a}_{22} = -2 \},
    \end{equation}
    where $\tilde{a}_{ij}$ is the $(i,j)$-th element of $\tilde{A}$.
    Then, let us show that
    1) the matrix $A$ cannot be uniquely identified
    from the pair $(\mathcal{D}, \setSigmapk)$
    but 2) $(\mathcal{D}, \setSigmapk)$ is jointly informative
    for the Lyapunov equation \eqref{eqn:Lyap}.
    
    First, we verify 1).
    From \eqref{eqn:def_setSigmaD}, $\setSigmaD$ can be expressed as
    $\setSigmaD = \{ \tilde{A} \in \setA_2 \mid
    \tilde{a}_{11} = -1, \ \tilde{a}_{21} = 0 \}$ 
    and thus we obtain
    \begin{equation}
      \label{eqn:ex/setSigmaDpk}
      \setSigmaD \cap \setSigmapk =
        \left\{
        \begin{bmatrix}
          -1 & \alpha \\
          0 & -2
        \end{bmatrix}
        \relmiddle{|} \alpha \in (0,2)
        \right\}.          
    \end{equation}
    Since  $|\setSigmaD \cap \setSigmapk | > 1$,
    the available information $(\mathcal{D}, \setSigmapk)$ does not uniquely identify the system.

    Next, 2) is shown as follows.
    By solving the Lyapunov equation
    \begin{align}
      \begin{bmatrix}
        -1 \!&\! \alpha \\
         0 \!&\! -2
      \end{bmatrix}\!\!\!
      \begin{bmatrix}
        p_{11} \!\!&\!\! p_{12} \\
        p_{21} \!\!&\!\! p_{22}
      \end{bmatrix}
      \!+\!      
      \begin{bmatrix}
        p_{11} \!\!&\!\! p_{12} \\
        p_{21} \!\!&\!\! p_{22}
      \end{bmatrix}\!\!\!
      \begin{bmatrix}
        -1 \!&\! \alpha \\
        0  \!&\! -2
      \end{bmatrix}^\top
      \!\!=\! -\!
      \begin{bmatrix}
        2 & 3 \\
        -3 & 0
      \end{bmatrix}
    \end{align}
    for each $\alpha \in (0,2)$, we obtain the solution
    \begin{equation}
      P = 
      \begin{bmatrix}
        p_{11} & p_{12} \\
        p_{21} & p_{22}
      \end{bmatrix}
      =
      \begin{bmatrix}
        1 & 1 \\
        -1 & 0
      \end{bmatrix}.
    \end{equation}
    Since $P$ does not depend on the parameter $\alpha \in (0,2)$,
    we can conclude that \eqref{eqn:Phi=Phi}
    holds for all 
    $\tilde{A}_1, \tilde{A}_2 \in \setSigmaD \cap \setSigmapk$.
    Thus,
    $(\mathcal{D}, \setSigmapk)$ is jointly informative for
    the Lyapunov equation \eqref{eqn:Lyap}. \thmend
  \end{example}
  
  In this example,
  both the dataset and the prior knowledge
  are required to uniquely determine $\Phi(A, Q)$
  when $A$ is unknown.
  The following example demonstrates this fact.
  
  \begin{example}
    \label{ex:only_data_or_pk}
    Consider $A$, $Q$,
    the dataset $\mathcal{D}$, and
    the prior knowledge $\setSigmapk$ in \exref{ex:motivating_example}.
    Then, let us verify that
    $\Phi(A,Q)$ cannot be uniquely determined from either
    1) the dataset $\mathcal{D}$ alone or 
    2) the prior knowledge $\setSigmapk$ alone.
    
    First, we verify 1) by showing that
    \eqref{eqn:Phi=Phi} is violated for some
    $\tilde{A}_1,\tilde{A}_2 \in \setSigmaD$.
    Let 
    \begin{equation}
      \tilde{A}_1 =
      \begin{bmatrix}
        -1 & 1 \\
         0 & -1
      \end{bmatrix}
    \end{equation}
    and $\tilde{A}_2 = A$.
    Then, $\tilde{A}_1, \tilde{A}_2 \in \setSigmaD$ holds.
    Since
    \begin{align}
      \Phi(\tilde{A}_1, Q) &=
        \begin{bmatrix}
          1 & 1.5 \\
          -1.5 & 0
        \end{bmatrix}, \\
        \Phi(\tilde{A}_2, Q) &= 
        \begin{bmatrix}
          1 & 1 \\
          -1 & 0
        \end{bmatrix},
        \label{eqn:ex/only_data/tildeA2}
    \end{align}
    we have
    \begin{equation}
      \label{eqn:Phi_neq_Phi}
      \Phi(\tilde{A}_1, Q) \neq \Phi(\tilde{A}_2, Q).
    \end{equation}
    This indicates that
    the dataset $\mathcal{D}$ alone does not contain sufficient information
    to uniquely determine the value of $\Phi(A,Q)$.
    
    Similarly, to verify 2),
    we show that
    \eqref{eqn:Phi=Phi} is violated for some
    $\tilde{A}_1,\tilde{A}_2 \in \setSigmapk$.
    Let
    \begin{equation}
      \tilde{A}_1 = 
      \begin{bmatrix}
        -2 & 0.5 \\
        4 & -2
      \end{bmatrix}
    \end{equation}
    and $\tilde{A}_2 = A$.
    Then, $\tilde{A}_1, \tilde{A}_2 \in \setSigmapk$, and 
    we have
    \begin{equation}
      \Phi(\tilde{A}_1, Q) =
      \begin{bmatrix}
        0.75 & 1.75 \\
        0.25 & 2
      \end{bmatrix}
    \end{equation}
    and \eqref{eqn:ex/only_data/tildeA2}.
    Hence, we again have \eqref{eqn:Phi_neq_Phi}.
    Thus, the prior knowledge $\setSigmapk$ alone
    is also insufficient to uniquely determine $\Phi(A,Q)$.
    
    These observations demonstrate that
    there exist problems for which computing $\Phi(A,Q)$
    from the given dataset alone is fundamentally impossible,
    regardless of the data-driven method employed.
    For such problems, with the dataset fixed,
    incorporating prior knowledge into the data-driven analysis
    is essential for uniquely determining $\Phi(A,Q)$.
    \thmend
  \end{example}
  
  Using this notion of joint informativity, 
  we formulate the following two problems.
  
  \begin{problem}
    \label{prob:informativity}
    Consider the system \eqref{eqn:sys} under $A \in \setA_n$.
    Suppose that a dataset $\mathcal{D}$ as in \eqref{eqn:def_dataset},
    prior knowledge $\setSigmapk \subseteq \setA_n$
    satisfying \eqref{eqn:A_in_setSigmapk},
    and a matrix $Q \in \setR^{n \times n}$ are given.
    Then, determine whether $(\mathcal{D}, \setSigmapk)$ is jointly informative for
    the Lyapunov equation \eqref{eqn:Lyap} or not. \thmend
  \end{problem}
  
  \begin{problem}
    \label{prob:computation}
    Consider the same setting as in \probref{prob:informativity}.
    Assume that $A$ is unknown but
     $(\mathcal{D}, \setSigmapk)$ is jointly informative for
    the Lyapunov equation \eqref{eqn:Lyap}.
    Compute $\Phi(A,Q)$. \thmend
  \end{problem}

  In \probref{prob:computation}, the unknown state matrix $A$
  is assumed to belong to $\setA_n$.
  This is precisely the necessary and sufficient condition
  for the Lyapunov equation \eqref{eqn:Lyap}
  to have a unique solution.
  The condition holds for almost all matrices
  $A \in \setR^{n \times n}$, since
  $\setR^{n \times n}\setminus\setA_n$
  has Lebesgue measure zero.
  Within this class, our framework characterizes when
  the unique solution can be determined from
  the available data and prior knowledge,
  without requiring unique identification of $A$.

\section{INFORMATIVITY ANALYSIS}
\label{sec:solution}
  This section provides solutions to
  Problems~\ref{prob:informativity}--\ref{prob:computation}.
  
\subsection{Characterization of Joint Informativity}

  In \probref{prob:informativity},
  no structural assumption 
  is imposed on the prior knowledge $\setSigmapk$
  except for the requirement $A \in \setSigmapk$,
  which makes it difficult to characterize a condition for the joint informativity.
  To overcome this difficulty, 
  the following lemma plays an important role.
  
  \begin{lemma}
    \label{lem:equivalence}
    Let a nonempty set $\setSigma \subseteq \setA_n$ 
    and a matrix $Q \in \setR^{n \times n}$ be given.
    Then, the following two statements are equivalent:
    \begin{enumerate}[(i)]
      \item \eqref{eqn:Phi=Phi} holds for all
        $\tilde{A}_1, \tilde{A}_2 \in \setSigma$.
      \item \eqref{eqn:Phi=Phi} holds for all
        $\tilde{A}_1, \tilde{A}_2 \in \affin(\setSigma) \cap \setA_n$.
        \thmend
    \end{enumerate}
  \end{lemma}
  
  \begin{proof}    
    The implication (ii) $\Rightarrow$ (i) follows from
    $\setSigma = \setSigma \cap \setA_n \subseteq \affin(\setSigma) \cap \setA_n$.    
    To prove the converse, assume that (i) holds.
    Choose $\bar{A}_0\in\setSigma$ and define
    $\bar{P}_0\coloneqq\Phi(\bar{A}_0,Q)$.
    We show that
    \begin{equation}
      \label{eqn:pf/lem-equiv/sufficient_condition}
      \forall \tilde{A} \in \affin(\setSigma) \cap \setA_n \ \ 
      \Phi(\tilde{A}, Q) = \bar{P}_0.
    \end{equation}
    
    Let $\tilde{A} \in \affin(\setSigma) \cap \setA_n$.
    By the definition of the affine hull,
    there exist $p\in\setZp$,
    matrices $\bar{A}_i\in\setSigma$,
    and scalars $c_i\in\setR$ $(i=1,\ldots,p)$ such that
    \begin{equation}
      \label{eqn:pf/lem-equiv/aff_combination}
      \tilde{A}=\sum_{i=1}^{p}c_i\bar{A}_i,
      \qquad
      \sum_{i=1}^{p}c_i=1.
    \end{equation}
    Since $\bar{A}_0, \bar{A}_i \in \setSigma$,
    statement (i) gives
    $\Phi(\bar{A}_i, Q) = \bar{P}_0$ for all $i = 1,2,\ldots,p$
    and thus we have
    \begin{equation}
      \label{eqn:pf/lem-equiv/Lyap_bAi}
      \bar{A}_i\bar{P}_0+\bar{P}_0\bar{A}_i^\top=-Q
      \quad (i=1,2,\ldots,p).
    \end{equation}
    It follows from \eqref{eqn:pf/lem-equiv/aff_combination} and \eqref{eqn:pf/lem-equiv/Lyap_bAi} that
    \begin{align}
      \tilde{A} \bar{P}_0 + \bar{P}_0 \tilde{A}^\top 
      &= \mleft( \sum_{i =1}^p c_i \bar{A}_i \mright) \bar{P}_0 + \bar{P}_0 \mleft( \sum_{j=1}^p c_j \bar{A}_j^\top \mright)
      \nonumber \\
      &= \sum_{i=1}^p c_i ( \bar{A}_i \bar{P}_0 + \bar{P}_0 \bar{A}_i^\top )
      \nonumber \\
      &= \sum_{i=1}^p c_i ( -Q ) = -Q.
    \end{align}
    Since $\tilde{A}\in\setA_n$, the corresponding
    Lyapunov equation has a unique solution.
    Therefore, $\Phi(\tilde{A},Q)=\bar{P}_0$.
    As $\tilde{A}$ was arbitrary,
    \eqref{eqn:pf/lem-equiv/sufficient_condition} holds.
    This proves (ii) and completes the proof.
  \end{proof}

  \lemref{lem:equivalence} provides us with an insight
  to understand the joint informativity for the Lyapunov equation \eqref{eqn:Lyap}.
  In fact, (i) in \lemref{lem:equivalence} coincides with
  the statement in \defref{def:JI_Lyap} when $\setSigma = \setSigmaD \cap \setSigmapk$.
  This immediately gives the following proposition.
  
  \begin{proposition}
    \label{prop:transformation}
    Consider \probref{prob:informativity}.
    Then, $(\mathcal{D}, \setSigmapk)$ is jointly informative
    for the Lyapunov equation \eqref{eqn:Lyap}
    if and only if
    \eqref{eqn:Phi=Phi} holds for all
    $\tilde{A}_1, \tilde{A}_2 \in \affin(\setSigmaD \cap \setSigmapk) \cap \setA_n$.
    \thmend
  \end{proposition}

\subsection{Necessary and Sufficient Condition for Informativity}
  \label{sec:informativity_theorem}
  
  Let us derive a solution to \probref{prob:informativity}. 
  
  \propref{prop:transformation} indicates that
  it is sufficient to consider the affine hull
  $\affin(\setSigmaD \cap \setSigmapk)$
  instead of $\setSigmaD \cap \setSigmapk$.
  This motivates us to introduce the following expression:
  \begin{equation}
    \label{eqn:structure_of_affine_setSigma}
    \affin(\setSigmaD \cap \setSigmapk)
    = A_0 + \vspan(\{A_1, A_2, \ldots, A_d\}),
  \end{equation}
  where $d \in \setZzp$ and $A_i \in \setR^{n \times n}$ ($i=0,1,2,\ldots, d$).
  For these $A_i$ ($i=0,1,2,\ldots, d$), we define
  the \textit{Lyapunov operator}
  $L_0: \setR^{n \times n} \to \setR^{n \times n}$ by
  \begin{equation}
    \label{eqn:def_L0}
    L_0(P) \coloneqq A_0 P + P A_0^\top
  \end{equation}
  and the set $\mathcal{K} \subseteq \setR^{n \times n}$ by
  \begin{equation}
    \label{eqn:def_setK}
    \mathcal{K} \coloneqq \mleft\{ P \in \setR^{n \times n} \relmiddle{|}
    A_i P + P A_i^\top = 0 \ \ (i = 1,2,\ldots,d) \mright\}.
  \end{equation}
  Then, a solution to \probref{prob:informativity}
  is given by the following theorem.
  
  \begin{theorem}
    \label{thm:informativity}
    Consider \probref{prob:informativity}.
    Let $A_i \in \setR^{n \times n}$ ($i=0,1,2,\ldots, d$)
    be matrices satisfying \eqref{eqn:structure_of_affine_setSigma}.
    Then, the following three statements are equivalent:
    \begin{enumerate}[(i)]
      \item $(\mathcal{D}, \setSigmapk)$ is jointly informative for the Lyapunov equation \eqref{eqn:Lyap}.
      \item The following equations have a solution $P \in \setR^{n \times n}$:
            \begin{equation}
              \label{eqn:informativity_Lyapeqns}
              \mleft\{
              \begin{aligned}
                A_0 P + PA_0^\top &= -Q, \\
                A_i P + PA_i^\top &= 0 \ \ (i=1,2,\ldots,d).
              \end{aligned}
              \mright.                
            \end{equation}
      \item The following condition holds:
            \begin{equation}
              \label{eqn:QinL0K}
              Q \in L_0(\mathcal{K}),
            \end{equation}
            where 
            $L_0(\mathcal{K}) \coloneqq
            \{ L_0(P) \in \setR^{n \times n} \mid P \in \mathcal{K} \}$.
            \thmend
    \end{enumerate}
  \end{theorem}
  
  \begin{proof}
    The equivalence of (ii) and (iii) is obvious
    from \eqref{eqn:def_L0}, \eqref{eqn:def_setK}, and the fact that $L_0(\mathcal{K})$ is a vector subspace.
    Therefore, we only provide a proof of (i) $\Leftrightarrow$ (ii).    
    From \propref{prop:transformation}, it is sufficient to show that
    \eqref{eqn:Phi=Phi} holds for all
    $\tilde{A}_1, \tilde{A}_2\in \affin(\setSigmaD \cap \setSigmapk) \cap \setA_n$
    if and only if \eqref{eqn:informativity_Lyapeqns} holds
    for some $P \in \setR^{n \times n}$.
    This statement is proved as follows.
    
    \underline{\textit{Proof of if part}}:
    Assume that $P^\ast \in \setR^{n \times n}$ is a solution to \eqref{eqn:informativity_Lyapeqns} and
    let $\tilde{A}_1, \tilde{A}_2 \in \affin(\setSigmaD \cap \setSigmapk) \cap \setA_n$.
    From \eqref{eqn:structure_of_affine_setSigma},
    there exist scalars $c_i \in \setR$ ($i=1,2,\ldots, d$) such that
    $\tilde{A}_1 = A_0 + \sum_{i=1}^d c_i A_i$.
    Therefore, we have
    \begin{align}
      \tilde{A}_1 P^\ast \!+\! P^\ast \tilde{A}_1^\top
        &=\! \mleft( \! A_0 \!+\! \sum_{i=1}^d c_i A_i \! \mright) P^\ast
          \!+\! P^\ast \mleft( \! A_0 \!+\! \sum_{j=1}^d c_j A_j \! \mright)^{\mkern-10mu\top}
          \nonumber \\
        &= (A_0 P^\ast \!+\! P^\ast A_0^\top) + \sum_{i=1}^d c_i (A_i P^\ast \!+\! P^\ast A_i^\top ) 
          \nonumber \\
        &= -Q
    \end{align}
    from \eqref{eqn:informativity_Lyapeqns},
    which implies $\Phi(\tilde{A}_1, Q) = P^\ast$.
    In the same way, we have $\Phi(\tilde{A}_2, Q) = P^\ast$.
    Thus, the if part is proved.
    
    \underline{\textit{Proof of only-if part}}:
    Assume
    \eqref{eqn:Phi=Phi} holds for all
    $\tilde{A}_1, \tilde{A}_2 \in \affin(\setSigmaD \cap \setSigmapk) \cap \setA_n$.
    By letting $P^\ast = \Phi(A, Q)$, we have
    \begin{equation}
      \label{eqn:informativity_cond2}
      \forall \tilde{A} \in \affin(\setSigmaD \cap \setSigmapk) \cap \setA_n \ \ 
      \Phi(\tilde{A}, Q) = P^\ast
    \end{equation}
    from the assumption and
    \begin{equation}
      \label{eqn:Lyap_Past}
      A P^\ast + P^\ast A^\top = -Q
    \end{equation}
    from the definition of $P^\ast$.
    Meanwhile, since $A \in \setA_n$ and $\setA_n$ is open,
    there exists $\epsilon \in \setRp$ such that
    $A+\epsilon A_i \in \setA_n$
    for all $i=1,2,\ldots,d$.
    Moreover, since $\affin(\setSigmaD \cap \setSigmapk)$ can also be expressed as
    \begin{equation}
      \affin(\setSigmaD \cap \setSigmapk)
      = A + \vspan(\{A_1, A_2, \ldots, A_d\}),
    \end{equation}
    it follows that
    $A + \epsilon A_i \in \affin(\setSigmaD \cap \setSigmapk)\cap \setA_n$ holds.
    This fact and \eqref{eqn:informativity_cond2} yield
    \begin{equation}
      \label{eqn:pf/informativity/Lyap_A+eAi}
      (A+\epsilon A_i ) P^\ast + P^\ast (A+\epsilon A_i)^\top = -Q \ \ 
      (i=1,2,\ldots,d).
    \end{equation}
    Here, \eqref{eqn:Lyap_Past} and \eqref{eqn:pf/informativity/Lyap_A+eAi}
    indicate the relation 
    \begin{equation}
      \label{eqn:pf/informativity/lyapAi}
      A_i P^\ast + P^\ast A_i^\top = 0 \ \ (i=1,2,\ldots,d).
    \end{equation}
    On the other hand, $A_0$ can be expressed as
    \begin{equation}
      \label{eqn:pf/informativity/A0}
      A_0 = A + \sum_{i=1}^d \theta_i A_i
    \end{equation}
    for some $\theta_i \in \setR$ ($i=1,2,\ldots,d$).
    By substituting \eqref{eqn:pf/informativity/A0} into $A_0 P^\ast + P^\ast A_0^\top$, we obtain
    \begin{align}
      A_0 P^\ast + P^\ast A_0^\top
      &= \mleft(\! A \!+\! \sum_{i=1}^d \theta_i A_i \! \mright) P^\ast
        \!+\! P^\ast \mleft( \!A \!+\! \sum_{i=1}^d \theta_i A_i \! \mright)^{\mkern-10mu\top}
      \nonumber \\
      &= (A P^\ast \!+\! P^\ast A^\top) + \sum_{i=1}^d \theta_i (A_i P^\ast \!+\! P^\ast A_i^\top ) 
      \nonumber \\
      &= -Q,
      \label{eqn:pf/informativity/LyapA0}
    \end{align}
    where the last equality is derived from \eqref{eqn:Lyap_Past} and \eqref{eqn:pf/informativity/lyapAi}.    
    This equation and \eqref{eqn:pf/informativity/lyapAi} 
    indicate the existence of $P \in \setR^{n \times n}$ satisfying \eqref{eqn:informativity_Lyapeqns},
    which proves the only-if part.
    Thus, the theorem is proved.
  \end{proof}
  
  In \thmref{thm:informativity},
  (ii) shows that our joint informativity is characterized by
  $d+1$ Lyapunov equations.
  On the other hand, (iii) is a closed form of the condition (ii)
  and is useful to characterize a class of Lyapunov-based system analysis
  that can be performed using only the available data and prior knowledge.

  One may consider that
  it is difficult to construct
  $\mathcal{K} \subseteq \setR^{n \times n}$
  directly from the definition in \eqref{eqn:def_setK}.
  In such a case, the following expression may be useful:
  \begin{equation}
    \label{eqn:setK_vec}
    \mathcal{K} = \mleft\{ P \in \setR^{n \times n} \relmiddle{|}
    \vec(P) \in \bigcap_{i=1}^d
    \kernel(I_n \!\otimes\! A_i \!+\! A_i \!\otimes\! I_n)\mright\},
  \end{equation}
  where $\vec(P) \in \setR^{n^2}$ denotes the vector obtained by stacking
  the columns of $P$ and
  $\otimes$ denotes the Kronecker product.

  \begin{example}
    \label{ex:informativity_theorem}
    Consider the same $A$, $Q$,
    the dataset $\mathcal{D}$, and
    the prior knowledge $\setSigmapk$
    in \exref{ex:motivating_example}.
    Then, let us verify that
    $(\mathcal{D}, \setSigmapk)$ is jointly informative for
    the Lyapunov equation \eqref{eqn:Lyap}
    by using \thmref{thm:informativity}.
    
    From \eqref{eqn:ex/setSigmaDpk}, 
    we have 
    \begin{equation}
      \label{ex/informativity_theorem/affSigma}
      \affin(\setSigmaD \cap \setSigmapk) = \mleft\{
      \begin{bmatrix}
        -1 & \alpha \\
        0 & -2
      \end{bmatrix}
      \relmiddle{|} \alpha \in \setR \mright\}.
    \end{equation}
    Then, let
    \begin{equation}
      A_0 =
      \begin{bmatrix}
        -1 & 0 \\
        0 & -2
      \end{bmatrix}, \ \ 
      A_1 = 
      \begin{bmatrix}
        0 & 1 \\
        0 & 0
      \end{bmatrix},
    \end{equation}
    for which \eqref{eqn:structure_of_affine_setSigma} is satisfied with $d=1$.
    From \eqref{eqn:def_setK}, $\mathcal{K} \subseteq \setR^{2 \times 2}$ is expressed as
    \begin{equation}
      \mathcal{K} = \mleft\{ 
      \begin{bmatrix}
        \beta & -\gamma \\
        \gamma & 0
      \end{bmatrix}
      \relmiddle{|} \beta, \gamma \in \setR \mright\}
    \end{equation}
    and thus $L_0(\mathcal{K}) \subseteq \setR^{2 \times 2}$ is obtained as
    \begin{equation}
      \label{eqn:ex/L0}
      L_0(\mathcal{K}) = \mleft\{ 
      \begin{bmatrix}
        -2\beta & 3\gamma \\
        -3\gamma & 0
      \end{bmatrix}
      \relmiddle{|} \beta, \gamma \in \setR \mright\}.
    \end{equation}
    From \eqref{eqn:ex/L0} and \eqref{eqn:ex/A_and_Q},
    we can find that \eqref{eqn:QinL0K} is satisfied.
    Thus, it follows from \thmref{thm:informativity} that
    $(\mathcal{D}, \setSigmapk)$ is jointly informative for
    the Lyapunov equation \eqref{eqn:Lyap}.
    \thmend
  \end{example}

  A solution to \probref{prob:computation}
  is obtained by the following theorem,
  which is a direct consequence of
  \defref{def:JI_Lyap}
  and \propref{prop:transformation}.
  
  \begin{theorem}
    \label{thm:computation}
    Consider \probref{prob:computation}.
    Let $\tilde{A} \in \affin(\setSigmaD \cap \setSigmapk) \cap \setA_n$.
    Then, the linear matrix equation
    \begin{equation}
      \label{eqn:Lyap_tildeA}
      \tilde{A} P + P \tilde{A}^\top = -Q
    \end{equation}
    has a unique solution $P^\ast \in \setR^{n \times n}$ and
    it is equal to $\Phi(A,Q)$.    
    \phantom{ }\thmend
  \end{theorem}

  \begin{example}
    \label{ex:computation_theorem}
    Consider the same $A$, $Q$,
    the dataset $\mathcal{D}$, and
    the prior knowledge $\setSigmapk$
    in \exref{ex:motivating_example}.
    Then, let us compute $\Phi(A, Q)$
    by using \thmref{thm:computation}.
    
    From \eqref{ex/informativity_theorem/affSigma},
    we have
    \begin{equation}
      \affin(\setSigmaD \cap \setSigmapk) \cap \setA_2 = \mleft\{
      \begin{bmatrix}
        -1 & \alpha \\
        0 & -2
      \end{bmatrix}
      \relmiddle{|} \alpha \in \setR \mright\}.      
    \end{equation}
    By choosing $\tilde{A} \in \affin(\setSigmaD \cap \setSigmapk) \cap \setA_2$ as
    \begin{equation}
      \tilde{A} = 
        \begin{bmatrix}
          -1 & 5 \\
          0 & -2
        \end{bmatrix}
    \end{equation}
    and solving \eqref{eqn:Lyap_tildeA},
    we obtain the unique solution 
    \begin{equation}
      P^\ast = 
      \begin{bmatrix}
        1 &  1 \\
        -1 & 0
      \end{bmatrix}.
    \end{equation}
    This result is consistent with \eqref{eqn:ex/PhiAQ}.
    Hence, we successfully computed the value of $\Phi(A,Q)$
    by \thmref{thm:computation}.
    \thmend
  \end{example}

  Once an affine representation satisfying \eqref{eqn:structure_of_affine_setSigma} is available, 
  the proposed test and computation reduce to solving linear equations. 
  Constructing this representation depends on
  the description of the prior-knowledge set
  and is not addressed here in full generality.

  It should be noted that  
  the proposed method is applicable
  to a strictly broader class of problems
  than methods based on unique system identification
  followed by model-based analysis
  and the direct data-driven method in \cite{bib:bannoLCSS2023}.
  In fact, Examples~\ref{ex:motivating_example}--\ref{ex:computation_theorem}
  illustrate a case where
  standard gray-box identification is impossible from
  $(\mathcal{D}, \setSigmapk)$
  (see 1) in \exref{ex:motivating_example}),
  whereas \thmref{thm:computation} can still be applied
  to obtain a solution, as demonstrated in \exref{ex:computation_theorem}.
  The method in \cite{bib:bannoLCSS2023} is also inapplicable
  to this problem because it requires
  $(\mathcal{D}, \setA_n)$ to be jointly informative
  for the Lyapunov equation \eqref{eqn:Lyap},
  which is not the case, as shown in 1) of \exref{ex:only_data_or_pk}.
  
\subsection{Special Case for Joint Informativity}
  \label{sec:special_case}

  Theorems~\ref{thm:informativity} and \ref{thm:computation}
  are based on
  $A_i$ ($i=0,1,2,\ldots, d$)
  satisfying \eqref{eqn:structure_of_affine_setSigma};
  however, they offer little structural interpretation 
  due to the general setting of $\setSigmapk$.  
  Therefore, this subsection provides further insight into our joint informativity
  by assuming the following structure for the prior knowledge $\setSigmapk$.
  
  \begin{assumption}
    \label{asm:special_structure_for_pk}
    There exist $A_{\rm r} \in \setSigmapk$, $l \in \setZp$, and $C \in \setR^{l \times n}$
    such that
    \begin{equation}
      \label{eqn:setSigmapk_sp}
      \setSigmapk = \{ \tilde{A} \in \setA_n \mid \exists H \in \setR^{n \times l} \ \ \tilde{A} = A_{\rm r} + HC \}
    \end{equation}
    holds.
    \thmend
  \end{assumption}
  
  This type of prior knowledge naturally arises
  when the system in \eqref{eqn:sys} includes
  unknown feedback or interactions
  that depend only on known features of the state.
  To illustrate this, consider a system in the form of
  \begin{equation}
    \mleft\{
      \begin{aligned}
        \dot{x}(t) &= A_{\rm r} x(t) + Hv(t) \\
        v(t) &= Cx(t),
      \end{aligned}
    \mright.
  \end{equation}
  where $A_{\rm r} \in \setR^{n \times n}$ and $C \in \setR^{l \times n}$ are known
  whereas the feedback or interaction gain $H \in \setR^{n \times l}$ is completely unknown. 
  The resulting state matrix is $A_{\rm r} + HC$, 
  and hence the set of all admissible state matrices
  is precisely the right-hand side of \eqref{eqn:setSigmapk_sp}. 
  Such a structure appears, for example, 
  in systems with unknown output feedback and 
  in networked systems with unknown interactions based on known relative-state variables.

  Let us reconsider \probref{prob:informativity}
  under \asmref{asm:special_structure_for_pk}.
  To characterize the set $\setSigmaD \cap \setSigmapk$,
  the following lemma is useful.
  
  \begin{lemma}
    \label{lem:setSigma_transformation}
    Consider \probref{prob:informativity}.
    Then, the following statements hold.
    \begin{enumerate}[(i)]
      \item The statement in \asmref{asm:special_structure_for_pk} holds
            if and only if
            there exist $m \in \setZp$ and $Y_0 \in \setR^{n \times m}$ such that
            \begin{equation}
              \label{eqn:setSigmapk_sp_Y0}
              \setSigmapk = \{\tilde{A} \in \setA_n \mid
              (\tilde{A} - A) Y_0 = 0 \}.
            \end{equation}
      \item Let $X_0 \in \setR^{n \times s}$ be a matrix satisfying
              \begin{equation}
                \image(X_0) = 
                \vspan\mleft( \bigcup_{t \in [0, \tf)} \{ x(t, x_0) \} \mright).
              \end{equation}
              Then, the equality
              \begin{equation}
                \label{eqn:setSigmaD_X0}
                \setSigmaD = \{ \tilde{A} \in \setA_n \mid (\tilde{A} - A) X_0 = 0 \}
              \end{equation}
              holds. \thmend
    \end{enumerate}
  \end{lemma}

  \begin{proof}
    Proofs are given in \asecref{asec:pf/setSigma_transformation}.
  \end{proof}
  
    From \eqref{eqn:setSigmapk_sp_Y0} and \eqref{eqn:setSigmaD_X0},
    $\setSigmaD \cap \setSigmapk$ is explicitly expressed as
    \begin{equation}
      \label{eqn:setSigma_sp}
      \setSigmaD \cap \setSigmapk = \{ \tilde{A} \in \setA_n \mid (\tilde{A} - A) Z = 0 \},
    \end{equation}
    where $Z \in \setR^{n \times r}$ is a matrix with full column rank satisfying
    $\image(Z) = \image([X_0 \ \ Y_0])$
    (note that $r = 0$ is allowed here).
    By considering the structure of $\setSigmaD \cap \setSigmapk$,
    the following result is obtained.

  \begin{theorem}
    \label{thm:informativity_sp}
    Consider \probref{prob:informativity}.
    Suppose that \asmref{asm:special_structure_for_pk} holds
    and $Z \in \setR^{n \times r}$ is a matrix
    with full column rank satisfying \eqref{eqn:setSigma_sp}.
    Let $A_0 \in \affin(\setSigmaD \cap \setSigmapk)$ be given.
    Then, the following three statements are equivalent:
    \begin{enumerate}[(i)]
      \item $(\mathcal{D}, \setSigmapk)$ is jointly informative for the Lyapunov equation \eqref{eqn:Lyap}.
      \item The following equations have a solution $P \in \setR^{n \times n}$:      
      \begin{equation}
        \label{eqn:informativity_sp_eqns}
        \mleft\{
        \begin{aligned}
          &A_0 P + PA_0^\top = -Q, \\
          &\Pi P = P\Pi = 0,
        \end{aligned}
        \mright.                
      \end{equation}
      where $\Pi \coloneqq I_n - ZZ^+$.
      \item \eqref{eqn:QinL0K} holds for
          \begin{equation}
            \mathcal{K} =
            \{ Z W Z^\top \mid W \in \setR^{r \times r} \}.
          \end{equation}
          \thmend
    \end{enumerate}
  \end{theorem}

  \begin{proof}
    Let us show that (i) $\Leftrightarrow$ (ii).    
    Let $d = n^2$, 
    $A_1 = E_{11}\Pi$, $A_2 = E_{12}\Pi$, $\ldots$, and $A_{n^2} = E_{nn}\Pi$,
    where $E_{ij} \in \setR^{n \times n}$ is the matrix
    whose $(i,j)$-th entry is one and all other entries are zero.
    Then, \thmref{thm:informativity} and the following facts prove
    the equivalence of (i) and (ii).
    
    \begin{enumerate}[(a)]
      \item \eqref{eqn:structure_of_affine_setSigma} holds
            for the given $A_0$ and the above $A_i$ $(i=1,2,\ldots,d)$.
      
      \item Let $P \in \setR^{n \times n}$. Then,
            \begin{equation}
              \label{eqn:pf/informativity_pf/LyapAi}
              A_i P + PA_i ^\top = 0 \ \ (i=1,2,\ldots,d)
            \end{equation}
            holds if and only if
            \begin{equation}
              \label{eqn:pf/informativity_pf/PPi}
              \Pi P = P\Pi = 0.
            \end{equation}
    \end{enumerate}
    
    First, we show (a).
    From \eqref{eqn:setSigma_sp} and \propref{prop:equiv} in \asecref{asec:pf/setSigma_transformation},
    $\setSigmaD \cap \setSigmapk$ is expressed as
    \begin{equation}
      \label{eqn:pf/informativity_sp/tildeA_parametrization}
      \setSigmaD \cap \setSigmapk = \{ \tilde{A} \in \setA_n \mid
        \exists K \in \setR^{n \times n} \ \  \tilde{A} = A + K \Pi  \}.
    \end{equation}
    Moreover, 
    \begin{equation}
      \label{eqn:pf/informativity_sp/affin}
      \affin(\setSigmaD \cap \setSigmapk)
      = A+\mathbf{V}
    \end{equation}
    holds, where
    $\mathbf{V} \coloneqq \{K\Pi \mid K \in \setR^{n \times n} \}$
    is a vector subspace of $\setR^{n \times n}$.
    This fact is shown as follows.
    Since $\setA_n$ is open and $A \in \setA_n$ holds,
    there exists $\epsilon>0$ such that
    \begin{equation}
      \mathbf{B}_{\epsilon}(A) \cap (A+\mathbf{V})
      \subseteq \setSigmaD \cap \setSigmapk
      \subseteq A+\mathbf{V},
    \end{equation}
    where 
    $\mathbf{B}_\epsilon(A) \coloneqq \{ \tilde{A} \in \setR^{n \times n} \mid
    \| \tilde{A} - A \|_{\rm F} < \epsilon \}$ is an open ball
    and $\| \tilde{A} - A \|_{\rm F}$ denotes the Frobenius norm of $\tilde{A} - A$.
    Since the affine hulls of the leftmost and the rightmost sets are both equal to $A+\mathbf{V}$,
    \eqref{eqn:pf/informativity_sp/affin} is deduced.
    From
    \eqref{eqn:pf/informativity_sp/affin} and
    $A_0 \in \affin(\setSigmaD \cap \setSigmapk)$,
    we conclude that $\affin(\setSigmaD \cap \setSigmapk) = A_0 + \mathbf{V}$.
    This fact and 
    $\mathbf{V} = \vspan(\{E_{ij}\Pi \mid i,j = 1,2,\ldots,n \})$ imply that
    \eqref{eqn:structure_of_affine_setSigma} holds.
    This proves (a).

    Next, let us show the if part of (b).
    From \eqref{eqn:pf/informativity_pf/PPi} and $\Pi^\top = \Pi$,
    we have
    \begin{equation}
      \label{eqn:pf/informativity_sp/Lyap_EijPi}
      (E_{ij}\Pi)P + P(E_{ij}\Pi)^\top = 0 \ \ (i,j=1,2,\ldots,n).
    \end{equation}
    These equalities and the definition of $A_i$ ($i=1,2,\ldots,d$)
    directly yield \eqref{eqn:pf/informativity_pf/LyapAi}.
    
    The only-if part of (b) is shown as follows.
    From \eqref{eqn:pf/informativity_pf/LyapAi} and the definition of $A_i$ ($i=1,2,\ldots,d$),
    we have \eqref{eqn:pf/informativity_sp/Lyap_EijPi}.
    Summing the equations in 
    \eqref{eqn:pf/informativity_sp/Lyap_EijPi} corresponding to $i=j$ yields
    \begin{equation}
      \label{eqn:pf/informativity_sp/PiP+PPi}
        \Pi P + P\Pi^\top = 0.
    \end{equation}
    Therefore, we obtain
    \begin{align}
      \Pi P
        &= \Pi (\Pi P) 
        = \Pi (-P \Pi^\top)
        = - \Pi P \Pi^\top \nonumber \\
        &= -\frac{1}{2} \Pi (\Pi P + P \Pi^\top) \Pi^\top
        = 0,
    \end{align}
    where the first and the fourth equalities follow from $\Pi^2 = \Pi$
    and the second and fifth equalities follow from \eqref{eqn:pf/informativity_sp/PiP+PPi}.
    Since $\Pi=\Pi^\top$,
    \eqref{eqn:pf/informativity_sp/PiP+PPi}
    and $\Pi P=0$ also give
    $P\Pi = P\Pi^\top= -\Pi P = 0$.
    Thus, \eqref{eqn:pf/informativity_pf/PPi} holds,
    which proves the only-if part.
    
    This concludes the proof of (i) $\Leftrightarrow$ (ii).
    
    Next, let us show (ii) $\Leftrightarrow$ (iii).
    From \thmref{thm:informativity}, \eqref{eqn:def_setK}, and
    the equivalence of \eqref{eqn:pf/informativity_pf/LyapAi}
    and \eqref{eqn:pf/informativity_pf/PPi},
    it is sufficient to show that
    \eqref{eqn:pf/informativity_pf/PPi} is equivalent to 
    \begin{equation}
      \label{eqn:pf/informativity_pf/existsW}
      \exists W \in \setR^{r \times r} \ \ P = Z W Z^\top.
    \end{equation}
    This is proved as follows.
    Assume \eqref{eqn:pf/informativity_pf/PPi}.
    The relation $\Pi P = 0$ provides $\image(P) \subseteq \kernel(\Pi)$.
    Meanwhile, the definition of $\Pi$ implies $\kernel(\Pi) = \image(Z)$ \cite{bib:bernstein2018}.
    Thus, we have $\image(P) \subseteq \image(Z)$.
    Similarly, we have $\image(P^\top) \subseteq \image(Z)$ from $P\Pi = 0$.
    These two facts indicate \eqref{eqn:pf/informativity_pf/existsW}.
    On the other hand, the converse follows by reversing the above argument.
    Thus, (ii) $\Leftrightarrow$ (iii) is proved.
  \end{proof}

  \thmref{thm:informativity_sp} indicates that
  $\Phi(A,Q)$, the unique solution to the Lyapunov equation in \eqref{eqn:Lyap},
  is parameterized as $\Phi(A,Q) = ZWZ^\top$
  for some $W \in \setR^{r \times r}$
  when $(\mathcal{D}, \setSigmapk)$ is jointly informative.  
  This fact follows from \thmref{thm:informativity_sp}
  with $A_0=A$
  and the equivalence of
  \eqref{eqn:pf/informativity_pf/PPi} and
  \eqref{eqn:pf/informativity_pf/existsW}.
  This parameterization can also be found in \cite{bib:bannoLCSS2023},
  which addresses \probref{prob:informativity}
  for the special case $\setSigmapk = \setA_n$.
  
  Furthermore, a solution to \probref{prob:computation}
  is obtained by the parameterization in \eqref{eqn:pf/informativity_pf/existsW}, as follows.
  
  \begin{theorem}
    \label{thm:computation_sp}
    Consider \probref{prob:computation}.
    Suppose that \asmref{asm:special_structure_for_pk} holds.
    Let $Z \in \setR^{n \times r}$ be a matrix with full column rank
    satisfying \eqref{eqn:setSigma_sp},
    and let $\tilde{A} \in \affin(\setSigmaD \cap \setSigmapk)$.
    Then, the linear matrix equation
    \begin{equation}
      \label{eqn:Lyap_tildeA_ZWZT}
      \tilde{A} Z W Z^\top + Z W Z^\top \tilde{A}^\top = -Q
    \end{equation}
    has a unique solution $W^\ast \in \setR^{r \times r}$ and
    \begin{equation}
      \label{eqn:PhiAQ_Wast}
      \Phi(A, Q) = ZW^\ast Z^\top.
    \end{equation}
    \thmend
  \end{theorem}

  \begin{proof}
    We first show that the solution set of
    \begin{equation}
      \label{eqn:Lyap_A_ZWZT}
      A Z W Z^\top + Z W Z^\top A^\top = -Q
    \end{equation}
    coincides with that of \eqref{eqn:Lyap_tildeA_ZWZT}.
    From $\tilde{A} \in \affin(\setSigmaD \cap \setSigmapk)$,
    \eqref{eqn:pf/informativity_sp/affin},
    and \propref{prop:equiv} in \asecref{asec:pf/setSigma_transformation},
    we obtain $(\tilde{A}- A) Z = 0$ and hence $\tilde{A}Z = AZ$.
    Therefore, for any $W \in \setR^{r \times r}$,
    \begin{align}
      \tilde{A}ZWZ^\top + ZWZ^\top\tilde{A}^\top
        &= (\tilde{A}Z)WZ^\top + ZW(\tilde{A}Z)^\top
        \nonumber\\
        &= AZWZ^\top + ZWZ^\top A^\top.
      \label{eqn:pf/computation_sp/operator_identity}
    \end{align}
    This equality directly shows that 
    \eqref{eqn:Lyap_A_ZWZT} and \eqref{eqn:Lyap_tildeA_ZWZT}
    have the same solution set.
    It therefore suffices to prove that
    \eqref{eqn:Lyap_A_ZWZT} has a unique solution $W^\ast$
    satisfying \eqref{eqn:PhiAQ_Wast}.
    
    First, we show existence.
    
    As stated in the discussion following \thmref{thm:informativity_sp},
    there exists $W^\ast \in \setR^{r \times r}$ satisfying \eqref{eqn:PhiAQ_Wast}.
    This~fact and the definition of $\Phi(A, Q)$                                                  %
    imply \eqref{eqn:Lyap_A_ZWZT} with $W = W^\ast$,
    which shows the existence of $W^\ast$ satisfying \eqref{eqn:Lyap_A_ZWZT} and
    \eqref{eqn:PhiAQ_Wast}.
    
    Next, we prove uniqueness.
    Let $W_1, W_2 \in \setR^{r \times r}$ be solutions to
    \eqref{eqn:Lyap_A_ZWZT}.
    Then, $P_1 \coloneqq ZW_1Z^\top$ and
    $P_2 \coloneqq ZW_2Z^\top$ are both solutions to
    \eqref{eqn:Lyap}.
    Since $A \in \setA_n$, this Lyapunov equation
    has a unique solution.
    Hence, $P_1=P_2$, which gives
    $Z(W_1-W_2)Z^\top=0$.
    Since $Z$ has full column rank, this implies $W_1=W_2$.
    Thus, the solution of \eqref{eqn:Lyap_A_ZWZT} is unique.
    
    By the equivalence established above,
    the same conclusion holds for
    \eqref{eqn:Lyap_tildeA_ZWZT}.
    This completes the proof.
  \end{proof}
    
  We provide two remarks for \thmref{thm:computation_sp}.
  First, unlike the general result in \thmref{thm:computation},
  \thmref{thm:computation_sp} does not require $\tilde{A} \in \setA_n$.
  This relaxation is beneficial because it allows
  $\tilde{A}$ to be chosen from the affine hull
  without imposing additional spectral constraints.
  Second, \thmref{thm:computation_sp} asserts that
  the computation of $\Phi(A,Q)$ is reduced to solving a linear equation with $r^2$ unknowns
  while \thmref{thm:computation} is based on solving 
  \eqref{eqn:Lyap_tildeA}, which has $n^2$ unknowns.
  This difference is useful for practical situations,
  e.g., when 
  \eqref{eqn:sys} is a large-scale system
  but $r$ is relatively small compared with the order $n$ of the system.

\section{CONCLUSIONS AND FUTURE WORK}
\label{sec:conclusion}

In this paper, we introduced a notion of joint informativity 
for a pair consisting of a dataset and prior knowledge, 
and provided a necessary and sufficient condition
for the pair to be jointly informative for the Lyapunov equation. 
We showed that the joint informativity is characterized
by the existence of a solution to a system of Lyapunov equations 
determined by the dataset and the prior knowledge. 
Furthermore, we explored joint informativity
in a special case where the prior knowledge has a specific structure.

The notion of joint informativity should be further investigated
for other data-driven tasks in system analysis and controller design. 
It would also be interesting to consider
experimental design problems in the presence of prior knowledge. 
In the context of Lyapunov equations, 
more general settings deserve further study. 
For example, this paper focuses on noise-free and full-state data; 
extensions to noisy or partial measurement data would be of interest.
  
\section*{ACKNOWLEDGMENT}
The author used ChatGPT by OpenAI
to assist with English-language editing,
literature review, and 
refinement of mathematical arguments throughout the manuscript.
The author takes full responsibility for the accuracy
and integrity of the manuscript.

\appendix
\subsection{Proof of \lemref{lem:setSigma_transformation}}
  \label{asec:pf/setSigma_transformation}
  
  Item (ii) is a direct consequence of
  Proposition 3 in \cite{bib:vanwaarde2019} and
  Lemma 2 in \cite{bib:bannoLCSS2023}.
  Therefore, we only provide the proof of (i).
  Our proof is based on the following fact.
  
  \begin{proposition}
    \label{prop:equiv}
    Let $M \in \setR^{n \times n}$, $Y_0 \in \setR^{n \times m}$, and $C \in \setR^{l \times n}$. 
    If
    \begin{equation}
      \label{eqn:imY0_kerC}
      \image(Y_0) = \kernel(C)
    \end{equation}
    holds, then
    $MY_0 = 0 \Leftrightarrow \exists H \in \setR^{n \times l} \ \ M = HC$.
    \thmend
  \end{proposition}
  
  \propref{prop:equiv} is a variant of Fact~3.13.12 in \cite{bib:bernstein2018}.
  In the following proof, we use
  $\setSigmapk^{\prime}$ and $\setSigmapk^{\prime\prime}$ to represent
  the right-hand sides of \eqref{eqn:setSigmapk_sp} and \eqref{eqn:setSigmapk_sp_Y0}, respectively.
  
  \underline{\textit{Proof of if part}}:
  Suppose that \eqref{eqn:setSigmapk_sp_Y0} holds
  for some $m$ and $Y_0$.
  We show that $\setSigmapk = \setSigmapk^{\prime}$ holds
  for $A_{\rm r} \coloneqq A$, $l \coloneqq n$, and $C \coloneqq I_n - Y_0 Y_0^{+}$
  (note that \eqref{eqn:imY0_kerC} holds).
  Let $\tilde{A} \in \setSigmapk$.
  Then, we have
  $(\tilde{A} - A_{\rm r})Y_0 = 0$
  from \eqref{eqn:setSigmapk_sp_Y0} and thus
  there exists
  $H \in \setR^{n \times n}$ satisfying $\tilde{A} - A_{\rm r} = HC$
  from \propref{prop:equiv}.
  Hence, we have $\setSigmapk \subseteq \setSigmapk^{\prime}$.
  On the other hand,
  $\setSigmapk^{\prime} \subseteq \setSigmapk$ is also shown by reversing the above procedure.
  Thus, we obtain $\setSigmapk = \setSigmapk^{\prime}$.

  \underline{\textit{Proof of only-if part}}:
  Suppose that \eqref{eqn:setSigmapk_sp} holds
  for some $A_{\rm r}$, $l$, and $C$.
  Let us show that $\setSigmapk = \setSigmapk^{\prime\prime}$ holds
  for $m \coloneqq n$ and $Y_0 \coloneqq I_n - C^{+}C$
  (\eqref{eqn:imY0_kerC} also holds in this case).
  Let $\tilde{A} \in \setSigmapk$.
  It follows from $A, \tilde{A} \in \setSigmapk$ and \eqref{eqn:setSigmapk_sp} that
  there exist $H_1, H_2 \in \setR^{n \times l}$ such that
  \begin{align}
    A &= A_{\rm r} + H_1 C, 
      \label{eqn:pf/setSigma_transformation/A=} \\
    \tilde{A} &= A_{\rm r} + H_2 C.
    \label{eqn:pf/setSigma_transformation/tildeA=}
  \end{align}
  By substituting \eqref{eqn:pf/setSigma_transformation/A=} into
  \eqref{eqn:pf/setSigma_transformation/tildeA=},
  we can show the existence of $H \in \setR^{n \times l}$ satisfying
  $\tilde{A} - A = HC$.
  This fact and \propref{prop:equiv}
  directly provide
  $(\tilde{A} - A)Y_0 = 0$,
  which implies $\tilde{A} \in \setSigmapk^{\prime\prime}$.
  Hence, $\setSigmapk \subseteq \setSigmapk^{\prime\prime}$ is obtained.
  Meanwhile, reversing this procedure with a slight modification
  also yields
  $\setSigmapk^{\prime\prime} \subseteq \setSigmapk$.
  Thus, $\setSigmapk = \setSigmapk^{\prime\prime}$ is proved.

\end{document}